\documentclass[11pt,letterpaper]{article}
\usepackage[margin=1in]{geometry}
\usepackage[T1]{fontenc}
\usepackage{lmodern}
\usepackage{amsmath,amssymb,amsthm,mathtools}
\usepackage{booktabs}
\usepackage[shortlabels]{enumitem}
\usepackage{microtype}
\usepackage{xcolor}
\usepackage{tikz}
\usetikzlibrary{patterns}
\usepackage[colorlinks=true,linkcolor=blue!50!black,citecolor=blue!50!black,urlcolor=blue!50!black]{hyperref}
\usepackage[nameinlink,capitalise]{cleveref}

\hypersetup{
  pdftitle={A QPTAS for Stochastic Scheduling of Bernoulli Jobs},
  pdfauthor={Junho Hwang}
}

\newtheorem{theorem}{Theorem}[section]
\newtheorem{lemma}[theorem]{Lemma}
\newtheorem{corollary}[theorem]{Corollary}

\theoremstyle{definition}

\theoremstyle{remark}

\newcommand{\OPT}{\mathrm{OPT}}
\newcommand{\E}{\mathbb{E}}
\DeclareMathOperator{\cost}{cost}

\title{A QPTAS for Stochastic Scheduling of Bernoulli Jobs}
\author{Junho Hwang}
\date{}

\begin{document}

\maketitle
\thispagestyle{empty}

\begin{abstract}
\normalsize
We study the classical problem of scheduling jobs with random processing times on $m$ identical machines to minimize the expected sum of completion times, for Bernoulli jobs: job $j$ takes time $p_j$ with probability $q_j$ and time $0$ otherwise, and its outcome is revealed when it starts.
The benchmark is an optimal adaptive policy.
We give a quasi-polynomial-time approximation scheme for every number of machines.
Previously, quasi-polynomial time was known to give an $O(\log N)$-approximation, and approximation schemes were known only for a constant number of distinct sizes.
The scheme rests on a simple observation: it suffices to round the times at which machines become free, rather than the times at which jobs start, to a grid whose width is proportional to the job size, and an optimal policy stretched by a factor close to one already respects such grids.
We also show that every policy that fixes the order of the jobs in advance loses a factor $\Omega(\log N)$, already on two machines, so a constant factor requires adapting the order to the observed outcomes.
Further results include a polynomial-time adaptive rule with ratio $\min\{m,1+\sum_jq_j\}$, a simpler approximation scheme for a constant number of sizes, and \#P-hardness of computing the optimal expected cost on two machines.

\end{abstract}

\newpage
\tableofcontents
\newpage

\section{Introduction}\label{sec:intro}

Scheduling jobs on identical parallel machines to minimize the sum of completion times is one of the most basic problems in scheduling.
With deterministic processing times it is solved by the shortest-processing-time rule~\cite{Smi56,BCS74}.
When processing times are random and only their distributions are known in advance, the scheduler is a \emph{policy}: it decides which job to start next based on what it has observed so far.
The benchmark is the best such policy, which may adapt its decisions to the observed processing times.
The problem has been studied for decades~\cite{Rot66,WP80,MRW84,WVW86}, and approximation algorithms are known for broad classes of distributions~\cite{MSU99,SU05,MUV06,SSU16,GMUX20,GMUX21}.
Their guarantees depend on an upper bound $\Delta$ on the squared coefficient of variation of the processing times, and they become weak when processing times have high variance.

A clean way to isolate the difficulty of high variance is to consider \emph{Bernoulli jobs}.
Job $j$ has processing time $p_j$ with probability~$q_j$ and $0$ otherwise, and the outcome becomes known when the job starts.
Such a job either finishes immediately or blocks its machine for a long time, and a good policy has to decide when to risk such a blocking.
Even for Bernoulli jobs, constant-factor approximations have been elusive.
For general distributions, Im, Moseley, and Pruhs~\cite{IMP15} gave a list policy with ratio $O(\log^2N+m\log N)$, and a constant factor with $(1+\varepsilon)m$ machines when $m=\Omega(\log N)$.
Eberle, Fischer, Matuschke, and Megow~\cite{EFMM19} showed that index policies lose a factor $\Omega(\Delta^{1/4})$, and gave an $O(m)$-approximation when the instance consists of identical deterministic jobs and identical Bernoulli jobs.
Gupta, Moseley, and Zhou~\cite{GMZ23} gave a list policy whose ratio is $\widetilde O(\sqrt m)$, where $\widetilde O$ hides factors polylogarithmic in~$N$.
Recently, Antoniadis, Hoeksma, Schewior, and Uetz~\cite{AHSU25} gave a polynomial-time approximation scheme (PTAS) when the number of distinct sizes $p_j$ is a constant.
They also gave a quasi-polynomial-time $O(\log N)$-approximation for arbitrary sizes, and left open whether constant factors are possible.
On the complexity side, Moseley, Pruhs, Uetz, and Zhou~\cite{MPUZ26} showed that computing the optimal expected cost, deciding whether it is at most a given value, and computing the expected cost of SEPT are \#P-hard for jobs with two-point distributions.

\subsection{Our results}

Throughout, $N$ is the number of jobs, $m<N$ the number of machines, $\mathcal L$ the input length, $P=\sum_jq_jp_j$ the expected total processing time, and $\OPT$ the expected sum of completion times of an optimal adaptive policy.
Our main result is an approximation scheme in quasi-polynomial time for every number of machines.

\begin{theorem}\label{thm:main}
For Bernoulli jobs on $m$ identical machines and every $\eta\in(0,1]$, a policy with expected sum of completion times at most $(1+\eta)\OPT$ can be computed in time $2^{O(\eta^{-2}\log^2(N/\eta))}\cdot\mathrm{poly}(\mathcal L)$.
Each decision of the policy is a table lookup.
\end{theorem}

This improves the quasi-polynomial $O(\log N)$-approximation of~\cite{AHSU25} to a ratio of $1+\eta$.
The same dynamic program, without the size reduction, gives a short proof of the PTAS of~\cite{AHSU25} for a constant number of sizes, with a running time whose exponent is linear in $1/\eta$.

\begin{theorem}\label{thm:fewsizes}
If the jobs have $n$ distinct sizes, a $(1+\eta)$-approximation can be computed in time $(Nm)^{O(n/\eta)}\cdot\mathrm{poly}(\mathcal L)$ for fixed $n$ and~$\eta$.
\end{theorem}

In polynomial time, we obtain a ratio that does not depend on~$N$.

\begin{theorem}\label{thm:mapprox}
A simple adaptive priority rule runs in polynomial time and has expected cost at most $\min\{m,1+\sum_jq_j\}\cdot\OPT$.
\end{theorem}

Earlier polynomial-time guarantees for Bernoulli jobs with arbitrary sizes grow with~$N$~\cite{IMP15,GMZ23}.
Theorem~\ref{thm:mapprox} gives a constant factor for every fixed number of machines, and the factor $2$ for two machines.
The second bound is useful when few jobs are expected to be positive.

Exact computation is hard already for two machines, while approximate evaluation of a given policy is easy.

\begin{theorem}\label{thm:hardintro}
For Bernoulli jobs on two machines, computing $\OPT$, deciding whether $\OPT\le x$, and computing the expected cost of the shortest-expected-processing-time rule (SEPT) are \#P-hard.
On the other hand, the expected cost of any given non-idling policy can be approximated within a factor $1\pm\rho$ with probability $1-\delta$ from $O(N\rho^{-2}\log(2/\delta))$ simulations.
\end{theorem}

The hardness part adapts the construction of~\cite{MPUZ26}, which uses jobs with two-point distributions on $m=N/2$ machines, to Bernoulli jobs on two machines.

Finally, we show that list policies cannot reach a constant factor.

\begin{theorem}\label{thm:barriersintro}
There are instances on two machines in which every policy that fixes the job order in advance costs $\Omega(\log N)\cdot\OPT$.
\end{theorem}

So the guarantees of~\cite{IMP15,GMZ23} for list policies must depend on~$N$, and the bound $O(\log^2N)$ of~\cite{IMP15} for two machines is tight up to a factor $O(\log N)$.

\subsection{Techniques}

\paragraph{Rounding release times instead of start times.}
A dynamic program over states of an optimal policy must record when each machine becomes free and which jobs remain.
Antoniadis, Hoeksma, Schewior, and Uetz~\cite{AHSU25} keep the number of states polynomial by transforming an optimal policy into a \emph{stratified} policy that starts the jobs of each type only at the points of a type-specific grid.
Their transformation loses a factor $1+O(n\varepsilon)$ for $n$ types~\cite[Theorem~4.23]{AHSU25}.
The number of states of their dynamic program grows as $m^{n\varepsilon^{-2z}}$ for some $1\le z\le n$, so the approach needs $n$ to be a constant.

We round only the times at which machines become free.
Fix an integer $K$ and let $\gamma=1+1/K$.
A positive job of size $p$ started at time~$s$ completes at $s+p$, and under our \emph{reuse rule} its machine becomes available again at the first point of the grid $(p/K)\mathbb Z$ that is at least $s+p$:
\[
  R_p(s)=\frac pK\Bigl(\Bigl\lceil\frac{sK}{p}\Bigr\rceil+K\Bigr)\in[s+p,\,s+\gamma p).
\]
The key observation is that an optimal policy stretched by the factor~$\gamma$ already obeys this rule.
If the stretched policy starts the job at $\gamma s$, then the next job on the same machine starts at least at $\gamma(s+p)>R_p(\gamma s)$.
Compressing the stretched policy back, by starting the same jobs in the same order as early as the rule allows, can only move start times earlier.
So the loss is the factor $\gamma$ on the expected sum of start times, for any number of sizes (Lemma~\ref{lem:stretch}).
Start times remain unconstrained, which is why no delays accumulate.
Figure~\ref{fig:stretch} shows the three policies in one realization.

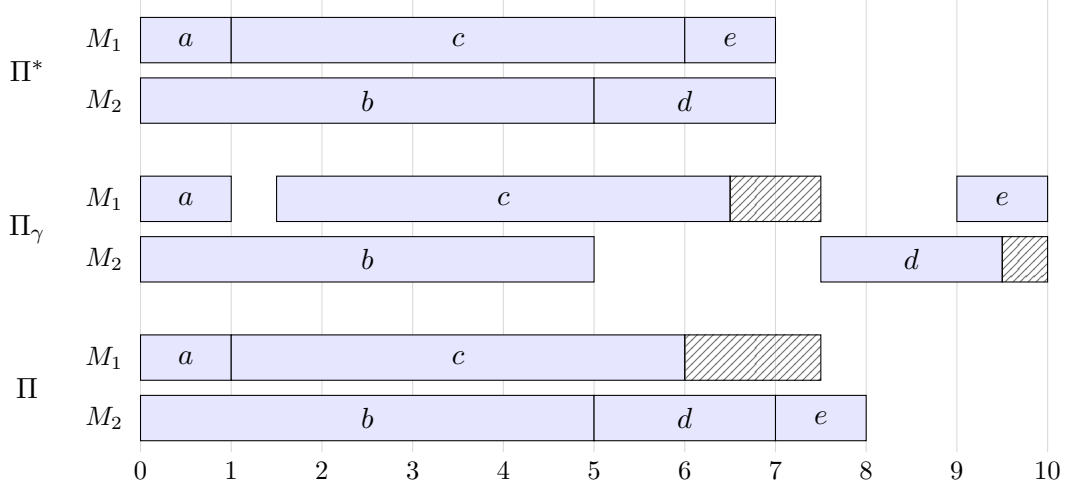
\begin{figure}[t]
\centering
\begin{tikzpicture}[x=1.2cm,y=0.5cm,
    job/.style={draw,fill=blue!10},
    resv/.style={draw,pattern=north east lines,pattern color=black!55}]
  \foreach \t in {0,...,10} {
    \draw[black!15] (\t,0.9) -- (\t,12.9);
    \node[below] at (\t,0.9) {\small$\t$};
  }
  \foreach \y/\name in {9.6/{\Pi^*}, 5.4/{\Pi_\gamma}, 1.2/{\Pi}} {
    \node at (-1.25,\y+1.4) {$\name$};
    \node[anchor=east] at (-0.08,\y+2.2) {\small$M_1$};
    \node[anchor=east] at (-0.08,\y+0.6) {\small$M_2$};
  }
  \foreach \s/\e/\y/\lab in {
      0/1/11.2/a, 1/6/11.2/c, 6/7/11.2/e, 0/5/9.6/b, 5/7/9.6/d,
      0/1/7.0/a, 1.5/6.5/7.0/c, 9/10/7.0/e, 0/5/5.4/b, 7.5/9.5/5.4/d,
      0/1/2.8/a, 1/6/2.8/c, 0/5/1.2/b, 5/7/1.2/d, 7/8/1.2/e} {
    \draw[job] (\s,\y) rectangle (\e,\y+1.2);
    \node at ({(\s+\e)/2},\y+0.6) {$\lab$};
  }
  \foreach \s/\e/\y in {6.5/7.5/7.0, 9.5/10/5.4, 6/7.5/2.8} {
    \draw[resv] (\s,\y) rectangle (\e,\y+1.2);
  }
\end{tikzpicture}
\caption{Stretching and compressing (Lemma~\ref{lem:stretch}) for $K=2$, so that $\gamma=3/2$, in a realization in which all five jobs are positive.
Jobs $a$ and~$e$ have size~$1$, jobs $b$ and~$c$ have size~$5$, and job~$d$ has size~$2$.
Top: a non-idling policy~$\Pi^*$ on two machines.
Middle: the stretched policy~$\Pi_\gamma$, which multiplies every start time by~$\gamma$ and keeps the machines; a job of size~$p$ started at time~$s$ releases its machine at $R_p(s)$, before the next start on that machine.
Bottom: the compressed policy~$\Pi$, which starts the jobs in the same order, each at the current minimum release time.
Hatched intervals are reserved by the reuse rule after a job has completed.
In~$\Pi$, job~$c$ starts at time~$1$, which is not a multiple of~$5/2$, so its machine stays reserved until $R_5(1)=15/2$.
Job~$e$ therefore starts at time~$7$ instead of~$6$, on the other machine, and still before its stretched start time~$9$.
The sums of the start times are $12$, $18$ and~$13$.}
\label{fig:stretch}
\end{figure}

The remaining work is to count the states of the dynamic program, which consist of the sorted vector of release times and the number of remaining jobs of each size class.
For the latter, it is enough to know how many jobs of each class remain: by an exchange argument of Gupta, Moseley, and Zhou~\cite{GMZ23}, some optimal policy starts the jobs of each size class in non-decreasing order of~$q_j$ (Lemma~\ref{lem:qorder}).
Two counting facts control the release times.
At a decision time~$t$, every release time is one of the $K+1$ grid points of its class in $[t,R_p(t)]$, so the vector of release times is a multiset of $m$ values drawn from $n(K+1)$ candidates.
The decision times themselves are compositions of at most $N$ maps $R_p$, and a block encoding of such compositions shows that there are at most $\binom{(K+1)N+n-1}{n}$ of them (Lemma~\ref{lem:times}).
For a constant number of sizes this gives Theorem~\ref{thm:fewsizes}.

For arbitrary sizes we first reduce to the case that all sizes lie within a polynomial ratio.
We build on the reduction of~\cite{GMZ23,AHSU25}, which tests very large jobs first and runs very small jobs at the start, but which loses a constant factor.
Our version keeps the approximation ratio up to an additive $2\varepsilon$ (Theorem~\ref{thm:reduction}).
It compares with the optimum through the inequality $\OPT\ge\OPT(U)+\sum_{j\notin U}q_jp_j$ for every set $U$ of jobs.
It bounds the cost of the rare event that a very large job is positive by Lemma~\ref{lem:freetime}, which bounds start times in terms of processed volume.
After rounding the sizes to powers of $1+\varepsilon$, the dynamic program has $n=O(\varepsilon^{-1}\log(N/\varepsilon))$ size classes, which gives Theorem~\ref{thm:main}.
After the reduction the number of decision times is polynomial, so the quasi-polynomial running time comes from two sources: the vectors of remaining counts over these $n$ classes, and the multisets of release times.

\paragraph{A greedy rule that minimizes a surrogate of the cost.}
For a set $S$ of started jobs, let $W_{m-1}(S)$ be their total realized processing time minus the $m-1$ largest realized values.
For non-idling policies, the start time of the $(k+1)$-st job lies between $W_{m-1}(S_k)/m$ and $W_{m-1}(S_k)$, where $S_k$ is the set of the first $k$ jobs.
So the sum of $W_{m-1}(S_k)$ over~$k$ approximates the cost within a factor~$m$.
This sum can be minimized exactly by an adaptive greedy rule.
Adding a realized value $x$ increases $W_{m-1}$ by $\min\{x,\theta\}$, where $\theta$ is the $(m-1)$-st largest realized value so far.
The rule starts a job of minimum expected increment $q_j\min\{p_j,\theta\}$, breaking ties by~$q_j$, and an exchange argument proves its optimality (Lemma~\ref{lem:greedyopt}).
The bound $1+\sum_jq_j$ comes from a different comparison between the sum and the start times.
It uses the fact that the $(m-1)$-st largest realized value never exceeds the second earliest time at which a machine becomes free.

\paragraph{Hardness and the list lower bound.}
The hardness proof follows the construction of~\cite{MPUZ26}.
It adds a job with a tiny probability whose size differs by $1/4$ between two instances, together with jobs of probability $1/2$ that encode the items of a \#Knapsack instance.
On these instances SEPT is optimal, by the theorem of Weber, Varaiya, and Walrand~\cite{WVW86}, which we extend to Bernoulli jobs by a smoothing argument (Lemma~\ref{lem:wvw}).
The difference of the two optimal costs then reveals the number of subsets of large weight.

The list lower bound uses types with $p_iq_i^2$ constant.
In every prefix of a list, pairs of positive jobs of the same type force a delay that is the same for each type, so every list accumulates delay $\Omega(\log N)$ per job.
A policy that tests jobs until the first positive one and then reorders the rest avoids this.

\paragraph{Organization.}
Section~\ref{sec:prelim} defines the model and collects basic facts.
Section~\ref{sec:reduction} gives the size reduction, and Section~\ref{sec:dp} gives the dynamic program and proves Theorems~\ref{thm:main} and~\ref{thm:fewsizes}.
Section~\ref{sec:greedy} proves Theorem~\ref{thm:mapprox}, Section~\ref{sec:hardness} proves Theorem~\ref{thm:hardintro}, and Section~\ref{sec:barriers} proves Theorem~\ref{thm:barriersintro}.
Section~\ref{sec:open} discusses open problems, including relaxations over start-time distributions.

\section{Model and preliminaries}\label{sec:prelim}

\paragraph{Jobs and policies.}
There are $N$ jobs and $m$ identical machines, all available at time~$0$.
Job~$j$ has an independent random processing time $X_j$ with
$\Pr[X_j=p_j]=q_j$ and $\Pr[X_j=0]=1-q_j$, where $p_j>0$ and $q_j\in(0,1]$ are rationals.
We call $p_j$ the \emph{size} of~$j$, and we call the outcome of~$j$ \emph{positive} if $X_j=p_j$ and \emph{zero} otherwise.
We write $\mu_j=\E[X_j]=q_jp_j$, $P=\sum_j\mu_j$ and $\Lambda=\sum_jq_j$.
Jobs of equal size form a \emph{size class}, and $n$ denotes the number of size classes.

The outcome of a job becomes known when the job starts: a positive job then occupies its machine for $p_j$ time units, and a zero job completes at once.
A \emph{policy} starts the jobs one at a time and without preemption.
Each decision chooses a remaining job, a machine, and a start time that is not earlier than the previous start time and at which the machine is not occupied.
The decision may depend on the outcomes of the jobs started before and on independent random coins, but not on the outcome of the chosen job.
In particular, several jobs may start at the same time, one after another, and each of these decisions sees the outcomes of the earlier ones.
Policies may leave machines idle on purpose.
For a policy~$\Pi$, let $S^\Pi_j$ and $C^\Pi_j=S^\Pi_j+X_j$ be the start and completion times of job~$j$.
The cost of~$\Pi$ is
\[
  \cost(\Pi)=\sum_j\E\bigl[C^\Pi_j\bigr]=P+\sum_j\E\bigl[S^\Pi_j\bigr],
\]
and $\OPT$ is the infimum of the cost over all policies.
A policy is an \emph{$\alpha$-approximation} if its cost is at most $\alpha\cdot\OPT$.
For a set $U$ of jobs, $\OPT(U)$ denotes the optimum of the instance consisting of the jobs in~$U$ alone on the same $m$ machines, and $P(U)=\sum_{j\in U}\mu_j$.

More generally, a \emph{state} consists of a set $U$ of remaining jobs, a current time $\tau$, and the times $\ell_1,\dots,\ell_m$ at which the machines become free.
Only the sorted vector of the times $\max\{\ell_i,\tau\}$ matters, and we call it the \emph{release vector} of the state.
The outcomes of the remaining jobs are independent of the past.
The cost of a policy from a state is the expected sum of the completion times of the jobs in~$U$.

\paragraph{Non-idling policies.}
A policy is \emph{non-idling} if it starts a job whenever a machine is free and jobs remain.
Equivalently, it starts every job at the minimum entry of the current release vector, on a machine that is free at that time.
Since the outcome of a job is known at its start, waiting reveals no information.
Hence an optimal policy need not idle~\cite[Lemma~2.1]{AHSU25}; we include the short argument.

\begin{lemma}\label{lem:nonidle}
From every state, some optimal policy is deterministic and non-idling.
In particular, $\OPT$ is a minimum over finitely many decision trees.
\end{lemma}

\begin{proof}
Let $V(\ell,U)$ be the infimum of the cost from a state with release vector $\ell$ and remaining jobs~$U$, with $V(\ell,\emptyset)=0$.
If $\ell\le\ell'$ coordinatewise, then $V(\ell,U)\le V(\ell',U)$: a policy for $\ell'$ can be run from $\ell$ with the same start times, using the machine of rank~$i$ in $\ell$ for the machine of rank~$i$ in~$\ell'$.
We show by induction on $|U|$ that $V(\ell,U)$ is attained by a deterministic non-idling policy.
Consider any policy and fix the outcome of its coins.
Its first decision starts a job $x$ at some time $s\ge\ell_1$ on a machine that is free at time~$s$.
Afterwards its release vector is obtained from the sorted vector of the values $\max\{\ell_i,s\}$, whose minimum entry is~$s$, by replacing that minimum entry by $s+X_x$.
Starting $x$ at time $\ell_1$ instead, on a machine with release time $\ell_1$, gives $x$ a completion time that is not larger, and the release vector is then obtained from $\ell$ by replacing its minimum entry $\ell_1$ by $\ell_1+X_x$.
We have $\ell_i\le\max\{\ell_i,s\}$ for every $i$ and $\ell_1+X_x\le s+X_x$.
Replacing the minimum entries of two coordinatewise dominated sorted vectors by values $x'\le y'$ and sorting again preserves the domination.
So the release vector after the shifted start is dominated by the one after the original start.
By monotonicity and induction, continuing with a deterministic non-idling optimal policy is no worse than the original continuation.
Averaging over the coins and choosing the best first job proves the claim for $U$.
There are finitely many deterministic non-idling decision trees, which gives the second statement.
\end{proof}

\paragraph{Basic bounds.}
The following two inequalities are used throughout.

\begin{lemma}\label{lem:basic}
\begin{enumerate}[label=(\alph*)]
\item $P\le\OPT$.
\item $\OPT\ge\OPT(U)+P(J\setminus U)$ for every set $U$ of jobs, where $J$ is the set of all jobs.
\end{enumerate}
\end{lemma}

\begin{proof}
Part~(a) holds because $C_j\ge X_j$ for every job.
For part~(b), let $\Pi$ be an optimal policy for~$J$.
Define a policy $\Pi_U$ for the jobs in~$U$ as follows.
It simulates~$\Pi$; whenever $\Pi$ starts a job $j\notin U$ on some machine, $\Pi_U$ draws an independent virtual outcome distributed as $X_j$ and keeps that machine idle for the virtual processing time.
This is a valid policy for~$U$, and the joint distribution of the start times and outcomes of the jobs in~$U$ is the same under~$\Pi_U$ as under~$\Pi$.
Hence $\OPT(U)\le\sum_{j\in U}\E[C^\Pi_j]$.
The jobs outside~$U$ contribute $\sum_{j\notin U}\E[C^\Pi_j]\ge P(J\setminus U)$, and adding the two bounds gives~(b).
\end{proof}

\paragraph{Free times.}
Fix a policy and a realization of the outcomes.
For $0\le k<N$, let $S_k$ be the set of the first $k$ jobs started and let $F_k$ be the start time of the $(k+1)$-st job.
For a set $S$ of started jobs and an integer $r\ge0$, let $W_r(S)$ be the total realized processing time of~$S$ minus the sum of the $r$ largest realized processing times in~$S$, with $W_r(S)=0$ if $|S|\le r$.
Adding a job to~$S$ never decreases $W_r(S)$: the sum gains the new value if it is not among the $r$ largest, and otherwise it gains the value that leaves the $r$ largest.
Also, $W_r(S)$ is non-increasing in~$r$.
Gupta, Moseley, and Zhou~\cite{GMZ23} measure the progress of a policy by its free times; the following lemma relates them to $W_{m-1}$.

\begin{lemma}\label{lem:freetime}
Let all machines be free at time~$0$, and consider a realization in which every job starts at the earliest time at which some machine is free.
This holds on every realization of a non-idling policy.
Then for every $0\le k<N$,
\[
  F_k\;\le\;W_{m-1}(S_k)\;\le\;m\,F_k.
\]
\end{lemma}

\begin{proof}
On such a realization every machine is busy without gaps from time~$0$ until it becomes free.
So after $k$ starts, the time at which machine~$i$ becomes free equals its load $L_i$, the total realized processing time of the jobs started on it, and $F_k=f:=\min_iL_i$.

The $m-1$ jobs with the largest realized processing times in $S_k$ lie on at most $m-1$ machines.
Some machine holds none of them, so its load is at most $W_{m-1}(S_k)$, and $f$ is at most that load.

For the second inequality, consider a machine with $L_i>f$ and let $J_i$ be the last job started on it.
The job $J_i$ is positive.
It started at the earliest free time of that moment, which is at most~$f$, since the earliest free time never decreases.
Hence $L_i-f\le X_{J_i}$.
At most $m-1$ machines have $L_i>f$, and their last jobs are distinct.
So the total realized processing time of $S_k$ is $\sum_iL_i\le mf+\sum_{i:L_i>f}X_{J_i}$, which is at most $mf$ plus the sum of the $m-1$ largest realized processing times in $S_k$.
This gives $W_{m-1}(S_k)\le mf$.
\end{proof}

\paragraph{Ordering within a size class.}
Fix an order of the jobs of each size class by non-decreasing $q_j$, breaking ties by index.
A policy \emph{respects the class orders} if, whenever it starts a job of some class, it starts the first remaining job of that class in this order.
Gupta, Moseley, and Zhou~\cite[Lemma~2.2]{GMZ23} state that an optimal policy starts the jobs of each size class in increasing order of~$q_j$.
Their proof in~\cite[Appendix~B]{GMZ23} is an exchange that moves a job with smaller $q_j$ before an equal-size job with larger $q_j$, using an independent coin, without increasing the cost.
We need the following existential form, with the tie-breaking used by our dynamic program, and we include the computation for the exchange.

\begin{lemma}\label{lem:qorder}
From every state, some optimal policy is deterministic, non-idling, and respects the class orders.
\end{lemma}

\begin{proof}
We use induction on the number of remaining jobs; the case of at most one job is trivial.
By Lemma~\ref{lem:nonidle}, take a deterministic non-idling optimal policy $T$, let $x$ be its first job, and let $t_0$ be its start time, the minimum entry of the release vector.
Let $a$ be the first remaining job of the class of $x$, so that $q_a\le q_x$, and let $p$ be their common size.
It suffices to find an optimal non-idling policy that starts $a$ first.
Indeed, after the outcome of $a$ is revealed we can then replace the continuation by a deterministic non-idling optimal policy for the resulting state that respects the class orders, which exists by induction.
If $a=x$ there is nothing to do, and if $q_a=q_x$, exchanging the names of $a$ and $x$ gives such a policy.

So let $q_a<q_x$.
The policy $T$ starts $a$ on every path.
Let $T_0$ and $T_1$ be its continuations after $x$ is zero and positive, and let $\sigma_0$ and $\sigma_1$ be the times at which they start~$a$.
These random times do not depend on the outcome of $a$, and $\sigma_0,\sigma_1\ge t_0$ since the minimum release time never decreases.
Let $\vartheta=(q_x-q_a)/((1-q_a)q_x)\in(0,1]$.
The policy $T'$ starts $a$ at time~$t_0$ on the machine used by~$T$ for~$x$.
If $a$ is positive, it continues as $T_1$, with $x$ started in place of~$a$.
If $a$ is zero, then with probability $1-\vartheta$ it continues as $T_0$, with $x$ started in place of~$a$.
With probability $\vartheta$ it instead starts $x$ at time $t_0$ and continues as $T_0$ or $T_1$, according to the outcome of $x$, treating the node of $a$ as a zero job.

The policy $T'$ is non-idling, and it continues as a copy of $T_1$ with probability $q_a+(1-q_a)\vartheta q_x=q_x$ and as a copy of $T_0$ with probability $1-q_x$.
Each copy is entered with the same release vector as in~$T$, because $a$ and $x$ have the same size.
In a copy of $T_1$, the job run at the node of $a$ is positive with probability $q_aq_x/q_x=q_a$.
In a copy of $T_0$, it is positive with probability $(1-q_a)(1-\vartheta)q_x/(1-q_x)=q_a$.
In both cases this outcome is independent of all other jobs.
So every job other than $a$ and $x$ has the same expected completion time in $T$ and~$T'$.
For $a$ and $x$, a direct computation from the branch probabilities above gives
\[
  \E_T[C_a+C_x]-\E_{T'}[C_a+C_x]\;=\;\frac{(1-q_x)(q_x-q_a)}{q_x}\bigl(\E[\sigma_0]-t_0\bigr)+(q_x-q_a)\bigl(\E[\sigma_1]-t_0\bigr)\;\ge\;0,
\]
where the terms with the factor $p$ cancel.
So $T'$ is an optimal non-idling policy that starts $a$ first.
\end{proof}

The stronger statement that reordering the jobs of each size class by~$q_j$ never increases the cost of a given policy, as stated in~\cite[Lemma~2.2]{AHSU25}, does not hold in general.
For instance, let $m=1$ and take jobs $a$ and $b$ of size~$1$ with $q_a=1/10$ and $q_b=9/10$, and a deterministic job~$d$ of size~$100$.
The policy that starts~$b$ first, then $d$ and $a$ if $b$ is zero, and $a$ and $d$ if $b$ is positive, costs $112.89$.
Its reordered version starts $a$ first and continues with the same tree, and costs $191.29$.
We use only the existential form.

\section{An approximation-preserving size reduction}\label{sec:reduction}

This section reduces the problem to instances whose sizes lie within a polynomial ratio, at a loss of $2\varepsilon$ in the approximation ratio.
We build on the reduction of Gupta, Moseley, and Zhou~\cite[Lemma~3.3 and Appendix~C]{GMZ23}, which tests very large jobs first and runs very small jobs at the start, at the cost of a constant factor.
Antoniadis, Hoeksma, Schewior, and Uetz~\cite[Lemma~A.1]{AHSU25} adapt it to turn an $\alpha$-approximation for instances whose sizes are powers of $c\ge2$ into a $(c^2+c\alpha)$-approximation.
The version below keeps the ratio $\alpha$ up to an additive $2\varepsilon$.
It uses Lemma~\ref{lem:basic}(b) to compare with the optimum, and Lemma~\ref{lem:freetime} to bound the cost of the rare event that a very large job is positive.

If $m\ge N$, starting every job at time~$0$ is optimal.
If $m=1$, the policy that starts the jobs in non-decreasing order of~$\mu_j$ (SEPT) is optimal~\cite{Rot66}; Section~\ref{sec:greedy} recalls the argument for adaptive policies.
We therefore assume $2\le m<N$.
Fix $\varepsilon\in(0,1]$ and set
\[
  T=\frac{NP}{\varepsilon},\qquad \sigma=\frac{\varepsilon P}{N^2},\qquad D=N\sigma=\frac{\varepsilon P}{N}.
\]
Call a job \emph{large} if $p_j\ge T$, \emph{small} if $p_j<\sigma$, and \emph{medium} otherwise, and let $L$, $Q$ and $M$ be the sets of large, small and medium jobs.
The sizes of the medium jobs lie in $[\sigma,T)$, so their ratio is less than $N^3/\varepsilon^2$.

Given a policy $A$ for the instance consisting of the medium jobs alone, the policy $\Pi_A$ for the whole instance works as follows.
\begin{enumerate}
\item At time~$0$, start the large jobs one after another on one machine, as long as their outcomes are zero.
\item If a large job is positive, schedule all remaining jobs from that moment on by list scheduling in any fixed order: whenever a machine is free, start the next job of the list on it.
\item Otherwise all large jobs are zero and have completed at time~$0$.
  List-schedule the small jobs from time~$0$, wait until time~$D$, and then run~$A$ on the medium jobs with every start time shifted by~$D$.
\end{enumerate}

\begin{theorem}\label{thm:reduction}
If $\cost(A)\le\alpha\cdot\OPT(M)$ for some $\alpha\ge1$, then $\cost(\Pi_A)\le(\alpha+2\varepsilon)\cdot\OPT$.
\end{theorem}

\begin{proof}
Let $E$ be the event that some large job is positive.

\emph{The event $E$.}
On~$E$, all decisions before the first positive large job occur at time~$0$, and the rest is list scheduling.
So on every realization in~$E$, every job starts at the earliest time at which some machine is free, and Lemma~\ref{lem:freetime} bounds every start time by $W_{m-1}(S_k)\le W_1(S_k)\le W_1(J)$, using $m\ge2$ and the monotonicity of~$W_r$.
If $\ell$ is a positive large job, then $W_1(J)$ is at most the total realized processing time minus~$X_\ell$.
Summing over all large jobs therefore gives, on every realization,
\[
  \mathbf 1_E\,W_1(J)\;\le\;\sum_{\ell\in L}\mathbf 1[X_\ell>0]\sum_{i\ne\ell}X_i .
\]
By independence, $\E[\mathbf 1[X_\ell>0]X_i]=q_\ell\mu_i$ for $i\ne\ell$.
Moreover $\sum_{\ell\in L}q_\ell\le\sum_{\ell\in L}\mu_\ell/T\le P/T$.
Hence $\E[\mathbf 1_EW_1(J)]\le P^2/T$, and
\[
  \E\Bigl[\mathbf 1_E\sum_jS_j\Bigr]\;\le\;N\cdot\frac{P^2}{T}\;=\;\varepsilon P .
\]

\emph{The complement of~$E$.}
On the complement, every large job starts at time~$0$.
The small jobs have total realized processing time at most $|Q|\sigma\le D$, so every machine has completed its small jobs by time~$D$, and all machines are free when $A$ starts.
Every medium job starts at time $D$ plus its start time under~$A$.
The event~$E$ depends only on the large jobs, so conditioning on its complement does not change the distribution of the other outcomes.
Therefore
\[
  \E\Bigl[\mathbf 1_{\overline E}\sum_{j\notin L}S_j\Bigr]
  \;\le\;(|Q|+|M|)\,D+\sum_{j\in M}\E\bigl[S^A_j\bigr]
  \;\le\;\varepsilon P+\cost(A)-P(M).
\]

\emph{Combining.}
Adding the two bounds,
\[
  \cost(\Pi_A)=P+\E\Bigl[\sum_jS_j\Bigr]
  \;\le\;P(L)+P(Q)+\cost(A)+2\varepsilon P
  \;\le\;P(L\cup Q)+\alpha\,\OPT(M)+2\varepsilon P .
\]
By Lemma~\ref{lem:basic}(b) with $U=M$ we have $\OPT(M)+P(L\cup Q)\le\OPT$, and by Lemma~\ref{lem:basic}(a) we have $P\le\OPT$.
Since $\alpha\ge1$, the right-hand side is at most $\alpha\,\OPT+2\varepsilon\,\OPT$.
\end{proof}

\section{A dynamic program over rounded release times}\label{sec:dp}

In this section the instance is arbitrary, with $n$ size classes and all machines free at time~$0$.
We fix an integer $K\ge1$ and set $\gamma=1+1/K$.
The dynamic program optimizes over policies whose machines are released only at points of size-dependent grids.
The main point is that stretching an optimal policy by the factor~$\gamma$ already makes it respect these grids, so the loss is a factor $\gamma$ on $\OPT-P$ regardless of the number of size classes.

\subsection{The reuse rule}

For a size $p$, let $g_p=p/K$ and define
\[
  R_p(s)=g_p\Bigl(\Bigl\lceil \frac{s}{g_p}\Bigr\rceil+K\Bigr)\qquad(s\ge0).
\]

\begin{lemma}\label{lem:reuse}
For every size $p$ and every $s\ge0$:
(a) $s+p\le R_p(s)<s+\gamma p$;
(b) $R_p$ is non-decreasing;
(c) $R_p(s)$ is an integer multiple of $g_p$.
\end{lemma}

\begin{proof}
We have $g_p\lceil s/g_p\rceil\in[s,s+g_p)$ and $Kg_p=p$, so $R_p(s)\in[s+p,s+p+g_p)=[s+p,s+\gamma p)$.
Parts (b) and (c) are immediate.
\end{proof}

A \emph{reuse-rule policy} is a policy in which a positive job of size $p$ started at time $s$ keeps its machine unavailable until time $R_p(s)$.
The job itself completes at time $s+p$, and the machine stays idle during $[s+p,R_p(s))$.
A zero job leaves its machine available at once.
The time from which a machine may be used again is its \emph{release time}.
A reuse-rule policy is \emph{grid-non-idling} if it always starts the next job at the current minimum release time, on a machine attaining it.

\subsection{Stretching and compressing an optimal policy}

\begin{lemma}\label{lem:stretch}
There is a grid-non-idling reuse-rule policy $\Pi$ that respects the class orders and satisfies $S^\Pi_j\le\gamma\,S^{\Pi^*}_j$ for every job $j$ and every realization, where $\Pi^*$ is a deterministic optimal policy as in Lemma~\ref{lem:qorder} and both policies see the same outcomes.
In particular, $\cost(\Pi)\le P+\gamma\,(\OPT-P)$.
\end{lemma}

\begin{proof}
Figure~\ref{fig:stretch} illustrates the proof in one realization.

\emph{Stretching.}
Let $\Pi_\gamma$ start every job on the same machine as~$\Pi^*$ at $\gamma$ times its start time in~$\Pi^*$.
At time $\gamma t$, the policy $\Pi_\gamma$ has observed exactly the outcomes that $\Pi^*$ has observed at time~$t$, so it can make the same decisions.
Suppose a job $i$ of size $p$ is positive, starts at time $s$ on some machine in~$\Pi^*$, and the next job on that machine starts at time $s'\ge s+p$.
In $\Pi_\gamma$ the machine is released at $R_p(\gamma s)<\gamma s+\gamma p\le\gamma s'$ by Lemma~\ref{lem:reuse}(a).
After a zero job, the next start on the machine is at $\gamma s'\ge\gamma s$.
So $\Pi_\gamma$ is a reuse-rule policy, possibly with idle time, with $S^{\Pi_\gamma}_j=\gamma S^{\Pi^*}_j$.

\emph{Compressing.}
The policy $\Pi^*$ makes its decisions in non-decreasing order of start time, so $\Pi_\gamma$ starts the jobs in the same order.
The policy~$\Pi$ starts the jobs in this order, each at its own current minimum release time and on a machine attaining it.
This is possible online: the $k$-th job of the order is determined by the outcomes of the first $k-1$ jobs, which $\Pi$ has already observed.
The order within each size class is that of~$\Pi^*$, so $\Pi$ respects the class orders.

For vectors $b\in\mathbb R^m$ and $\tau\ge0$, let $b\vee\tau$ be the vector $(\max\{b_1,\tau\},\dots,\max\{b_m,\tau\})$ sorted increasingly.
Let $a$ and $b$ be the sorted release vectors of $\Pi$ and $\Pi_\gamma$ just before the $k$-th start, and let $\tau_k$ be the start time of the $k$-th job in~$\Pi_\gamma$.
We show by induction that $a\le b\vee\tau_k$ coordinatewise.
Initially $a=b=0$.
The policy $\Pi_\gamma$ starts the $k$-th job at~$\tau_k$ on a machine released by~$\tau_k$, so the corresponding entry of $b\vee\tau_k$ equals $\tau_k$, the minimum of $b\vee\tau_k$.
Hence $\Pi$ starts the job at time $s_k=a_1\le\tau_k$.
If the job is zero, neither release vector changes.
If it is positive with size $p$, the new release vector of $\Pi$ is obtained from $a$ by replacing its minimum entry by $R_p(s_k)$ and sorting.
Likewise, since $R_p(\tau_k)\ge\tau_k$, the vector $b'\vee\tau_k$ for the new release vector $b'$ of $\Pi_\gamma$ is obtained from $b\vee\tau_k$ by replacing the entry $\tau_k$ of the used machine by $R_p(\tau_k)$ and sorting.
Since $R_p(s_k)\le R_p(\tau_k)$ by Lemma~\ref{lem:reuse}(b), the new vectors satisfy the same coordinatewise domination: replacing the minimum entries of two dominated sorted vectors by values $x\le y$ and sorting again preserves domination.
Finally $\tau_{k+1}\ge\tau_k$, so $b'\vee\tau_{k+1}\ge b'\vee\tau_k$ coordinatewise.
This completes the induction, and $S^\Pi_j=s_k\le\tau_k=\gamma S^{\Pi^*}_j$ for the $k$-th job~$j$.

Summing over the jobs, $\cost(\Pi)=P+\sum_j\E[S^\Pi_j]\le P+\gamma\sum_j\E[S^{\Pi^*}_j]=P+\gamma(\OPT-P)$.
\end{proof}

\subsection{The dynamic program}

Number the size classes $1,\dots,n$, let $p_c$ be the size of class~$c$, and let $N_c$ be the number of jobs in it.
We optimize over grid-non-idling reuse-rule policies that respect the class orders of Section~\ref{sec:prelim}, which sort each class by $q_j$ and break ties by index.
A state is a pair $(\ell,\nu)$.
Here $\ell=(\ell_1\le\dots\le\ell_m)$ is the sorted vector of release times, in absolute time.
The vector $\nu=(\nu_1,\dots,\nu_n)$ counts the remaining jobs of each class; the remaining jobs of class~$c$ are the last $\nu_c$ jobs of its order.
For $\nu_c>0$, let $j_c(\nu)$ be the first remaining job of class~$c$, and let $\ell^{(c)}$ be the vector obtained from~$\ell$ by replacing $\ell_1$ by $R_{p_c}(\ell_1)$ and sorting.
Let $e_c$ denote the $c$-th unit vector.
Define $V(\ell,0)=0$ and, for $\nu\ne0$,
\[
  V(\ell,\nu)=\min_{c:\,\nu_c>0}\Bigl[\ell_1+q_jp_c+(1-q_j)\,V(\ell,\nu-e_c)+q_j\,V\bigl(\ell^{(c)},\nu-e_c\bigr)\Bigr],
  \qquad j=j_c(\nu).
\]

\begin{lemma}\label{lem:bellman}
$V\bigl((0,\dots,0),(N_1,\dots,N_n)\bigr)$ is the minimum cost of a grid-non-idling reuse-rule policy that respects the class orders.
A policy attaining it chooses, in each state, a class $c$ attaining the minimum and starts~$j_c(\nu)$.
\end{lemma}

\begin{proof}
In such a policy, the next job starts at time $\ell_1$, so the expected completion time of the job started in state $(\ell,\nu)$ is $\ell_1+q_jp_c$.
After the start, the state is $(\ell,\nu-e_c)$ if the job is zero and $(\ell^{(c)},\nu-e_c)$ if it is positive.
The outcomes of the remaining jobs are independent of the past, the costs add up, and the machines are identical.
So the minimum expected cost of the remaining jobs depends only on the state and satisfies the recursion.
The recursion is acyclic because $\sum_c\nu_c$ decreases.
\end{proof}

\subsection{Counting the states}

\begin{lemma}\label{lem:times}
All release times that occur in the dynamic program lie in a set of size at most $\binom{(K+1)N+n-1}{n}$.
\end{lemma}

\begin{proof}
Write $g_c=g_{p_c}$ and $R_c=R_{p_c}$.
A release time is either $0$ or equal to $R_c(t)$ for a decision time $t$ at which a positive job of class~$c$ starts.
The time $t$ is itself $0$ or an earlier release time.
Following this chain backwards shows that every release time has the form $v=R_{c_k}\circ\dots\circ R_{c_1}(0)$, where $k\le N$ because the maps correspond to distinct positive jobs.
We encode such a composition by a vector in $\mathbb Z_{\ge0}^n$.

Split the sequence $c_1,\dots,c_k$ into blocks.
The first block ends at the last occurrence of the class with the largest grid width in the sequence, and the remaining suffix is split in the same way.
The largest grid widths of successive blocks strictly decrease.
Consider a block of length~$r$ whose largest grid width is $g$, which maps the value $u$ to the value~$v$.
Each map $R_{c'}$ in the block has $g_{c'}\le g$ and increases its argument by an amount in $[Kg_{c'},(K+1)g_{c'})$, by Lemma~\ref{lem:reuse}(a).
Hence $Kg\le v-u<(K+1)rg$.
The last map of the block has grid width~$g$, so $v$ is a multiple of~$g$, and we can write $v=g(\lceil u/g\rceil+a)$ for an integer~$a$.
Since $g\lceil u/g\rceil\in[u,u+g)$, the two bounds on $v-u$ give $K\le a\le(K+1)r-1$.

The encoding records the value~$a$ of each block at the coordinate of the block's largest class, and $0$ at every other coordinate.
Processing the nonzero coordinates in decreasing order of grid width, starting from $u=0$, reconstructs~$v$.
With $\kappa\ge1$ blocks, the coordinates sum to at most $(K+1)k-\kappa\le(K+1)N-1$.
The number of vectors in $\mathbb Z^n_{\ge0}$ with coordinate sum at most $(K+1)N-1$ is $\binom{(K+1)N+n-1}{n}$.
\end{proof}

For example, let $K=1$, so that $R_p(s)=p(\lceil s/p\rceil+1)$ and the grid width of each size is the size itself, and let the sizes be $4$, $2$ and~$1$.
Applying the maps $R_1,R_4,R_2,R_4,R_1,R_1$ in this order to~$0$ produces the values $1,8,10,16,17,18$.
The first block $R_1,R_4,R_2,R_4$ ends at the last occurrence of the largest size and maps $0$ to $16=4(\lceil0/4\rceil+4)$, so its value of~$a$ is~$4$.
The second block $R_1,R_1$ maps $16$ to $18=1\cdot(\lceil16/1\rceil+2)$, so its value of~$a$ is~$2$.
With the coordinates ordered by decreasing size, the encoding is $(4,0,2)$, and its coordinate sum $6$ is at most $(K+1)k-\kappa=10$.
Decoding computes $16$ and then~$18$ by the same two formulas, without the intermediate values $1$, $8$, $10$ and~$17$.
The maps $R_4,R_4,R_4,R_4,R_1,R_1$ have the same encoding and the same value.

\begin{lemma}\label{lem:profiles}
For every decision time $t$, the number of sorted release vectors $\ell$ with $\ell_1=t$ that occur in the dynamic program is at most $\binom{m+d-1}{d-1}$, where $d=n(K+1)$.
\end{lemma}

\begin{proof}
Consider a machine whose last positive job belongs to class~$c$ and started at time $s\le t$.
Its release time satisfies $t\le\ell_i=R_c(s)\le R_c(t)=g_c(\lceil t/g_c\rceil+K)$, and $\ell_i$ is a multiple of~$g_c$.
Hence $\ell_i\in\{g_c(\lceil t/g_c\rceil+a):0\le a\le K\}$.
A machine without a positive job has release time~$0$, which happens only if $t=0$ and is then contained in this set with $a=0$.
So every entry of~$\ell$ takes one of at most $d$ values, and a sorted vector of $m$ such entries is a multiset of size~$m$ drawn from $d$ values.
\end{proof}

\begin{theorem}\label{thm:dp}
For every instance with $n$ size classes and every integer $K\ge1$, a policy $\Pi$ with
\[
  \cost(\Pi)\le P+\Bigl(1+\frac1K\Bigr)(\OPT-P)
\]
can be computed in time
\[
  \binom{(K+1)N+n-1}{n}\binom{m+n(K+1)-1}{n(K+1)-1}\prod_{c=1}^n(N_c+1)\cdot\mathrm{poly}(N,m,\mathcal L),
\]
where $\mathcal L$ is the input length.
Each decision of $\Pi$ is then a table lookup.
\end{theorem}

\begin{proof}
By Lemma~\ref{lem:stretch}, some policy in the class optimized by the dynamic program has cost at most $P+\gamma(\OPT-P)$.
By Lemma~\ref{lem:bellman}, the dynamic program finds a policy in this class of minimum cost.
By Lemmas~\ref{lem:times} and~\ref{lem:profiles}, the number of states is at most the product of the three factors.
Each transition takes $O(n+m)$ arithmetic operations.
All release times are multiples of $1/(K\cdot\mathrm{lcm})$, where $\mathrm{lcm}$ is the least common multiple of the size denominators.
All values of $V$ are rationals whose denominators divide the product of $K\cdot\mathrm{lcm}$ and the denominators of all~$q_j$.
So the numbers have polynomial bit length.
\end{proof}

\begin{corollary}\label{cor:fewsizes}
For every fixed number $n$ of size classes and every fixed $\eta\in(0,1]$, a $(1+\eta)$-approximation can be computed in time $(Nm)^{O(n/\eta)}\mathrm{poly}(\mathcal L)$.
\end{corollary}

\begin{proof}
Take $K=\lceil1/\eta\rceil$ in Theorem~\ref{thm:dp}, so that $\cost(\Pi)\le P+(1+\eta)(\OPT-P)\le(1+\eta)\OPT$.
The three factors are at most $((K+1)N+n)^n$, $(m+n(K+1))^{n(K+1)}$ and $(N+1)^n$, and their product is $(Nm)^{O(n/\eta)}$ for fixed $n$ and~$\eta$.
This proves Theorem~\ref{thm:fewsizes}.
\end{proof}

For comparison, the dynamic program of Antoniadis, Hoeksma, Schewior, and Uetz has $O(N^{2n}\varepsilon^{-2n}m^{n\varepsilon^{-2z}})$ states, where $1\le z\le n$ depends on how the sizes group~\cite[Lemma~6.8]{AHSU25}.

\subsection{Proof of Theorem~\ref{thm:main}}\label{sec:qptas-proof}

Let $\eta\in(0,1]$.
Replacing $\eta$ by $1/\lceil1/\eta\rceil\ge\eta/2$, we may assume that $1/\eta$ is an integer; then all rounded sizes below have bit length polynomial in $\mathcal L$ and~$1/\eta$.
Let $\varepsilon=\eta/5$.
If $m=1$ or $m\ge N$, an optimal policy is known (Section~\ref{sec:reduction}), so assume $2\le m<N$.
Define the medium jobs as in Section~\ref{sec:reduction}.
Round the size of every medium job $j$ up to $p'_j=\sigma(1+\varepsilon)^{\lceil\log_{1+\varepsilon}(p_j/\sigma)\rceil}$, so that $p_j\le p'_j<(1+\varepsilon)p_j$, and keep $q_j$.
Call the result the rounded medium instance.
Its sizes lie in $[\sigma,(1+\varepsilon)T)$, so it has $n\le\log_{1+\varepsilon}(N^3/\varepsilon^2)+2=O(\varepsilon^{-1}\log(N/\varepsilon))$ size classes.

\emph{Rounding loses a factor $1+\varepsilon$.}
Take an optimal policy for the medium jobs and run it on the rounded medium instance with all start times multiplied by $1+\varepsilon$.
A positive job that occupies $[s,s+p_j)$ in the original policy occupies $[(1+\varepsilon)s,(1+\varepsilon)s+p'_j)$, which is contained in $[(1+\varepsilon)s,(1+\varepsilon)(s+p_j))$.
So the stretched policy is valid, and each of its completion times is at most $1+\varepsilon$ times the original one.
Hence the optimum of the rounded medium instance is at most $(1+\varepsilon)\OPT(M)$.

\emph{The dynamic program.}
Apply Theorem~\ref{thm:dp} to the rounded medium instance with $K=\lceil1/\varepsilon\rceil$.
This gives a policy $A$ whose cost on the rounded instance is at most $(1+1/K)$ times its optimum, which is at most $(1+\varepsilon)^2\OPT(M)$.
Run $A$ on the original medium jobs and keep each machine unavailable until the release time computed from the rounded size.
The start times are the same as on the rounded instance, and every completion time is at most the corresponding one there, since $p_j\le p'_j$.
So the cost of $A$ on the medium jobs is at most $(1+\varepsilon)^2\OPT(M)$.

\emph{The ratio.}
By Theorem~\ref{thm:reduction} with $\alpha=(1+\varepsilon)^2$, the policy $\Pi_A$ has cost at most $((1+\varepsilon)^2+2\varepsilon)\OPT\le(1+5\varepsilon)\OPT=(1+\eta)\OPT$, using $\varepsilon\le1$.

\emph{The running time.}
We have $K+1\le2/\varepsilon+1$ and $d=n(K+1)=O(\varepsilon^{-2}\log(N/\varepsilon))$.
After the size reduction, the number of possible release times is polynomial.
Every release time is a multiple of some $g_c=p'_c/K$ and at most $\gamma Np'_{\max}$, and $p'_{\max}/p'_c<(1+\varepsilon)N^3/\varepsilon^2$.
So there are at most $1+\sum_c(\gamma NKp'_{\max}/p'_c+1)=O(nKN^4/\varepsilon^2)$ release times, and this bound can replace the first factor in Theorem~\ref{thm:dp}.
The second factor is at most $(m+d)^d\le(N+d)^d=\exp(O(\varepsilon^{-2}\log^2(N/\varepsilon)))$, since $m<N$.
The third is at most $(N+1)^n=\exp(O(\varepsilon^{-1}\log N\log(N/\varepsilon)))$.
The rounding and the reduction take polynomial time.
So the total running time is $2^{O(\eta^{-2}\log^2(N/\eta))}\mathrm{poly}(\mathcal L)$.
\qed

\section{A polynomial-time \texorpdfstring{$m$}{m}-approximation}\label{sec:greedy}

This section analyzes a simple adaptive priority rule.
Fix distinct ranks $\rho(j)$ for the jobs.
For $m\ge2$, the policy $G$ maintains the multiset $A$ of the $m-1$ largest realized processing times among the jobs started so far, padded with zeros, and its minimum $\theta=\min A$.
Whenever a machine is free, $G$ starts a remaining job that minimizes, in lexicographic order,
\[
  \bigl(q_j\min\{p_j,\theta\},\;q_j,\;\rho(j)\bigr).
\]
After each start $G$ updates $A$.
The policy is non-idling: whenever a machine is free, for instance right after a zero outcome, it starts the next job.
For $m=1$, $G$ is SEPT.
Initially $\theta=0$, so $G$ first tries the jobs with the smallest~$q_j$.
Each decision scans the remaining jobs, so $G$ runs in $O(N^2)$ arithmetic operations.

\begin{theorem}\label{thm:greedy}
For every instance, $\cost(G)-P\le\min\{m,\,1+\Lambda\}\cdot(\OPT-P)$, where $\Lambda=\sum_jq_j$.
In particular, $\cost(G)\le\min\{m,1+\Lambda\}\cdot\OPT$.
\end{theorem}

For $m=1$, consider a non-idling policy that starts the jobs $J_1,\dots,J_N$ in this order, where $J_k$ may depend on earlier outcomes.
The processing time of $J_k$ adds to its own completion time and to those of all later jobs, and $J_k$ is chosen before its outcome is revealed.
So the cost is $\sum_k(N-k+1)\E[\mu_{J_k}]$, which SEPT minimizes~\cite{Rot66}; by Lemma~\ref{lem:nonidle} this proves the theorem for $m=1$.
From now on assume $m\ge2$ and write $r=m-1$.
The analysis compares policies through the quantity
\[
  H(\Pi)=\E\Bigl[\sum_{k=0}^{N-1}W_r(S_k)\Bigr],
\]
which depends only on the adaptive order in which $\Pi$ starts the jobs.
For every policy $\sum_jS_j=\sum_kF_k$, and for a non-idling policy Lemma~\ref{lem:freetime} applies on every realization.
So for every non-idling policy $\Pi$,
\begin{equation}\label{eq:sandwich}
  \sum_j\E\bigl[S^\Pi_j\bigr]\;\le\;H(\Pi)\;\le\;m\sum_j\E\bigl[S^\Pi_j\bigr].
\end{equation}

\subsection{The greedy rule minimizes \texorpdfstring{$H$}{H}}

Consider the increments $D_k=W_r(S_{k+1})-W_r(S_k)$.
If the multiset of the $r$ largest realized processing times so far has minimum $\theta$, then a new realized value $x$ changes $W_r$ by $\min\{x,\theta\}$.
So the expected increment of job~$j$ is $d_j(\theta)=q_j\min\{p_j,\theta\}$.
Since $W_r(S_0)=0$, we have $\sum_{k=0}^{N-1}W_r(S_k)=\sum_{k=0}^{N-1}(N-1-k)D_k$.
For the remaining decisions, the relevant state is the set $R$ of remaining jobs together with the multiset~$A$.
If $|R|=s$ and the next $s$ increments are $D'_1,\dots,D'_s$, the future part of the objective is $\sum_{i=1}^s(s-i)D'_i$.

\begin{lemma}\label{lem:greedyopt}
For every multiset $A$ and every set $R$ of remaining jobs, the greedy rule minimizes $\E[\sum_{i=1}^s(s-i)D'_i]$ over all adaptive orders.
\end{lemma}

\begin{proof}
We use induction on $s=|R|$; the case $s\le1$ is trivial.
The objective equals $(s-1)D'_1$ plus the future part of the objective from the state after the first decision, which has $s-1$ remaining jobs.
So by induction, every order may continue greedily after its first decision without becoming worse.
Let $a$ be the job chosen by the greedy rule, let $\theta=\min A$, and write $\mu_b=q_bp_b$.
Take an optimal order that is greedy after its first decision, and suppose it starts with a job $j\ne a$.
We construct an order that starts with~$a$ and is no worse.

\emph{Case 1: $p_a\le\theta$.}
Then $d_a(\theta')=\mu_a$ for every $\theta'\ge\theta$, while $d_b(\theta')$ is non-decreasing in $\theta'$ for every other job~$b$.
The threshold never decreases, so after $j$ the job $a$ is still the greedy choice.
The optimal order thus starts with $j,a$.
Swap them to $a,j$.
In both orders, $a$ contributes the increment $X_a$, because $p_a\le\theta$ is at most the threshold at any later time, and $j$ contributes $\min\{X_j,\theta\}$, because $a$ does not change~$A$.
Both orders end in the same state, and the swap changes the expected objective by $\mu_a-d_j(\theta)\le0$.

\emph{Case 2: $p_a>\theta$ and $p_j\le\theta$.}
Now $j$ does not change~$A$, so after $j$ the job~$a$ is again the greedy choice.
Swapping $j,a$ to $a,j$ leads to the same state and changes the expected objective by $q_a\theta-\mu_j=d_a(\theta)-d_j(\theta)\le0$.

\emph{Case 3: $p_a>\theta$ and $p_j>\theta$.}
First, $q_a=\min_bq_b$ and $q_a\theta\le\mu_b$ for every remaining job~$b$.
If $\theta>0$, a job $b$ with $q_b<q_a$ would have $d_b(\theta)\le q_b\theta<q_a\theta=d_a(\theta)$, contradicting the choice of~$a$; if $\theta=0$, the second key selects a job of minimum~$q$.
Moreover $q_a\theta=d_a(\theta)\le d_b(\theta)\le\mu_b$.

Let $\theta_2$ be the second smallest element of~$A$, with $\theta_2=\infty$ if $r=1$.
After a positive outcome of~$a$, the new threshold is $\theta_a=\min\{p_a,\theta_2\}$; after a positive outcome of $j$ it is $\theta_j=\min\{p_j,\theta_2\}$.
Let $u=\min\{p_a,p_j,\theta_2\}$ and
\[
  \mathcal B=\{b\in R\setminus\{a,j\}:\mu_b<q_au\}.
\]
Every $b\in\mathcal B$ has $q_b\ge q_a$ and therefore $p_b=\mu_b/q_b<u$.
So after $a$ or $j$ is positive, the jobs of~$\mathcal B$ do not change~$A$, and each of them contributes the increment~$X_b$.

We first determine the optimal order, which is greedy after~$j$.
If $j$ is positive, the threshold is $\theta_j\ge u$.
Every $b\in\mathcal B$ then has key $d_b(\theta_j)=\mu_b<q_au$, and $a$ has key $q_a\min\{p_a,\theta_j\}=q_au$.
Every other job $c$ has key $\min\{\mu_c,q_c\theta_j\}\ge q_au$.
If a job $c\notin\mathcal B$ has key exactly $q_au$ and $q_c=q_a$, then $\mu_c\ge q_au\ge q_a\theta$, so $c$ and $a$ had equal first and second keys at the start, and the rank favours~$a$.
So the order processes $\mathcal B$ in the order of the keys $(\mu_b,q_b,\rho(b))$, during which $A$ does not change, and then~$a$.
If $j$ is zero, the order starts~$a$ next.
If $a$ is then positive, the same argument with $\theta_a\ge u$ shows that it processes~$\mathcal B$ next, in the same order.

The new order starts with~$a$.
If $a$ is positive, it processes $\mathcal B$ in the same order and then starts~$j$.
If $a$ is zero, it starts~$j$, and if $j$ is positive it then processes~$\mathcal B$ in the same order.
In each case it then continues as the optimal order does from the same state.
Writing $X_{\mathcal B}$ for the sequence of outcomes of the jobs in $\mathcal B$ in this order, the increments are as follows.
\begin{center}
\begin{tabular}{lll}
\toprule
outcomes of $a$ and $j$ & optimal order & new order\\
\midrule
both zero & $0,0$ & $0,0$\\
$a$ positive, $j$ zero & $0,\theta,X_{\mathcal B}$ & $\theta,X_{\mathcal B},0$\\
$a$ zero, $j$ positive & $\theta,X_{\mathcal B},0$ & $0,\theta,X_{\mathcal B}$\\
both positive & $\theta,X_{\mathcal B},u$ & $\theta,X_{\mathcal B},u$\\
\bottomrule
\end{tabular}
\end{center}
In every row the two segments have the same length and end in the same state, because the final multiset $A$ does not depend on the order of insertion.
The weights of consecutive positions differ by exactly one.
So the new order is worse by $\theta+\sum_{b\in\mathcal B}X_b$ in the second row and better by the same amount in the third row.
The outcomes of the jobs in $\mathcal B$ are independent of those of $a$ and~$j$.
So the expected change is
\[
  \bigl(q_a(1-q_j)-(1-q_a)q_j\bigr)\Bigl(\theta+\sum_{b\in\mathcal B}\mu_b\Bigr)=(q_a-q_j)\Bigl(\theta+\sum_{b\in\mathcal B}\mu_b\Bigr)\le0. \qedhere
\]
\end{proof}

\subsection{Proof of Theorem~\ref{thm:greedy}}

The bound with factor $m$ follows from Lemma~\ref{lem:greedyopt} and~\eqref{eq:sandwich}.
Let $\Pi^*$ be a non-idling optimal policy (Lemma~\ref{lem:nonidle}).
Then $\cost(G)-P=\sum_j\E[S^G_j]\le H(G)\le H(\Pi^*)\le m(\OPT-P)$.

For the factor $1+\Lambda$ it suffices, by the same chain, to show that $H(\Pi)\le(1+\Lambda)\sum_j\E[S^\Pi_j]$ for every non-idling policy~$\Pi$.
Fix such a policy and a realization.
Before the $(k+1)$-st start, let $\ell_1(k)\le\ell_2(k)$ be the two smallest times at which machines become free.
Then $F_k=\ell_1(k)$.
Let $J_k$ be the job started at that position, $\chi_k\in\{0,1\}$ its outcome indicator, and $\theta_k$ the $(m-1)$-st largest realized processing time among the first $k$ started jobs, or $0$ if there are fewer than $m-1$ positive ones.

We first show that $\theta_k\le\ell_2(k)$.
Each of the first $k$ jobs started at a time at most $F_{k-1}\le\ell_1(k)\le\ell_2(k)$.
So a started job with realized processing time $x>\ell_2(k)$ is still running at time $\ell_2(k)$, since it started at a nonnegative time.
It therefore occupies a machine that becomes free only after time $\ell_2(k)$.
Jobs running at the same time occupy distinct machines, and at most $m-2$ machines become free after time $\ell_2(k)$.
So at most $m-2$ realized processing times exceed $\ell_2(k)$, which gives the claim.

If $J_k$ is zero, neither $W_r$ nor the earliest free time changes.
If it is positive with size $p$, then $W_r(S_{k+1})-W_r(S_k)=\min\{p,\theta_k\}$ and $F_{k+1}-F_k=\min\{p,\ell_2(k)-\ell_1(k)\}$.
Since $\theta_k\le\ell_2(k)=\ell_1(k)+(\ell_2(k)-\ell_1(k))$ and $\ell_1(k)\ge0$, in both cases
\[
  W_r(S_{k+1})-W_r(S_k)\;\le\;F_{k+1}-F_k+\chi_kF_k .
\]
Here $F_N$ denotes the earliest free time after the last start; it receives weight zero below.
Both $W_r(S_0)$ and $F_0$ are zero.
Multiplying the inequality by $N-1-k$ and summing over $k$ gives
\[
  \sum_{k=0}^{N-1}W_r(S_k)\;\le\;\sum_{k=0}^{N-1}F_k+\sum_{k=0}^{N-1}(N-1-k)\chi_kF_k .
\]
The choice of $J_k$ and the value $F_k$ are determined before the outcome of $J_k$ is revealed, so $\E[\chi_k\mid\text{history}]=q_{J_k}$.
On every realization $F_l\ge F_k$ for $l>k$, hence
\[
  \sum_kq_{J_k}(N-1-k)F_k\;\le\;\sum_kq_{J_k}\sum_{l>k}F_l\;=\;\sum_lF_l\sum_{k<l}q_{J_k}\;\le\;\Lambda\sum_lF_l .
\]
Taking expectations, $H(\Pi)\le(1+\Lambda)\,\E[\sum_kF_k]=(1+\Lambda)\sum_j\E[S^\Pi_j]$.
Finally, $\cost(G)-P\le\min\{m,1+\Lambda\}(\OPT-P)\le\min\{m,1+\Lambda\}\OPT-P$, which gives the second claim.
\qed

\paragraph{The factor $m$ is tight for $G$.}
Take $m-1$ jobs of size $K$ with $q=1-K^{-3}$ and $K$ deterministic jobs of size~$1$.
The rule $G$ starts the long jobs first, since $\theta=0$ and they have the smaller~$q$.
With probability close to one they occupy $m-1$ machines until time~$K$, while the unit jobs run on the remaining machine.
So $\cost(G)\ge(1-o(1))K^2/2$, whereas running the unit jobs first on all machines costs $K^2/(2m)+O(mK)$.

\section{Exact evaluation is hard, approximate evaluation is easy}\label{sec:hardness}

Moseley, Pruhs, Uetz, and Zhou~\cite{MPUZ26} show that computing the optimal expected cost, deciding whether it is at most a given value, and computing the expected cost of SEPT (and of its weighted version) are \#P-hard for processing times with two-point distributions.
Their instances have one machine for each of the $n$ items of the knapsack instance and $2n$ jobs, so $m=N/2$.
They combine jobs that encode the items of a knapsack instance with blocker jobs that are positive with small probability, show that an optimal policy starts the blockers first, and then use the theorem of Weber, Varaiya, and Walrand~\cite{WVW86}.
We adapt this construction to Bernoulli jobs on two machines.
We also show that the expected cost of a given non-idling policy can nevertheless be approximated efficiently.

\subsection{SEPT for stochastically ordered Bernoulli jobs}

Recall that $X\le_{\mathrm{st}}Y$ means $\Pr[X>x]\le\Pr[Y>x]$ for all~$x$.
Weber, Varaiya, and Walrand~\cite{WVW86} prove that if the processing times are stochastically ordered, then SEPT minimizes the expected sum of completion times on identical machines that become available at arbitrary given times.
They assume, for convenience, that the processing times have densities, and their induction considers strategies that start a job whenever a machine becomes available.
The next lemma states the version we need, for Bernoulli jobs, in the model of Section~\ref{sec:prelim} with machine~$i$ unavailable before a given time~$t_i$.
We take from~\cite{WVW86} only an exchange inequality for list policies and redo the induction, using Lemma~\ref{lem:nonidle}.

\begin{lemma}\label{lem:wvw}
Let $X_1\le_{\mathrm{st}}\dots\le_{\mathrm{st}}X_N$ be Bernoulli jobs, and let the machines become available at given times $t_1,\dots,t_m\ge0$.
The list policy that starts the jobs in the order $1,\dots,N$, each at the earliest time at which some machine is available, minimizes $\sum_j\E[C_j]$ among all policies.
\end{lemma}

\begin{proof}
For a list $\lambda$ of jobs and a vector $\tau$ of availability times, let $c(\tau;\lambda)$ be the expected cost of the list policy~$\lambda$ from~$\tau$.
For $k>1$, let $L^*_k$ be the list of the jobs in $\{2,\dots,N\}\setminus\{k\}$ in increasing order.

\emph{An exchange inequality.}
For every $\tau$ and every $k>1$,
\begin{equation}\label{eq:wvwexchange}
  c\bigl(\tau;(k,1)+L^*_k\bigr)\;\ge\;c\bigl(\tau;(1,k)+L^*_k\bigr),
\end{equation}
where $+$ denotes concatenation.
Weber, Varaiya, and Walrand prove~\eqref{eq:wvwexchange} for stochastically ordered processing times with densities and every vector of availability times, in the proof of their theorem~\cite[pp.~845--846]{WVW86}.
In their notation the inequality reads $R(\tau;L_{1,k})\ge R(\tau;L_{k,1})$ for the reward $r(t)=-t$, and its proof uses only their Lemmas~1--3.
For Bernoulli jobs, let $\delta>0$, let $U_1,\dots,U_N$ be independent and uniform on $[0,1]$, and let $Y_j=X_j+\delta U_j$.
The variables $Y_j$ have densities, and $Y_1\le_{\mathrm{st}}\dots\le_{\mathrm{st}}Y_N$, because adding independent copies of the same variable preserves the stochastic order.
So~\eqref{eq:wvwexchange} holds for the processing times~$Y_j$.
For a fixed list, couple the two runs through $Y_j=X_j+\delta U_j$.
At each start, the coordinatewise distance between the two sorted vectors of availability times grows by at most~$\delta$, since sorting is $1$-Lipschitz in the maximum norm.
So every start time changes by at most $N\delta$, and every list cost changes by at most $N(N+1)\delta$.
Letting $\delta\to0$ proves~\eqref{eq:wvwexchange} for the jobs~$X_j$.

\emph{Optimality.}
We show by induction on $N$ that, for every $\tau$, the list $(1,\dots,N)$ is optimal from~$\tau$; the case $N=1$ is trivial.
By Lemma~\ref{lem:nonidle}, some optimal policy from $\tau$ is deterministic and non-idling; let $k$ be its first job, started at the minimum entry of~$\tau$.
After the outcome of $k$ is revealed, the continuation is a policy for the other jobs from the resulting availability times.
All of them are at least the start time of $k$, so the continuation problem is an instance of the lemma with fewer jobs, and by induction the list of the other jobs in increasing order is optimal there.
So the list $(k)$ followed by the other jobs in increasing order is optimal from~$\tau$.
If $k=1$, this is the list $(1,\dots,N)$.
Otherwise it is $(k,1)+L^*_k$, which by~\eqref{eq:wvwexchange} is no better than $(1,k)+L^*_k$.
After job~$1$ starts, the list $(k)+L^*_k$ is a policy for the jobs $2,\dots,N$, so by induction the list $(2,\dots,N)$ is no worse.
Hence $(1,\dots,N)$ is optimal.
\end{proof}

\subsection{Hardness of exact computation}

\begin{theorem}\label{thm:hard}
For Bernoulli jobs on $m=2$ machines, each of the following problems is \#P-hard under polynomial-time Turing reductions:
computing $\OPT$; deciding whether $\OPT\le x$ for a given rational~$x$; and computing the expected cost of SEPT.
This holds even when one job has small~$q$, one job is deterministic, all other jobs have $q=1/2$, and all sizes are integers.
\end{theorem}

We reduce from \#Knapsack: given positive integers $a_1,\dots,a_k$ and an integer $B_0$ with $0\le B_0<S=\sum_ia_i$, count the subsets $I\subseteq[k]$ with $\sum_{i\in I}a_i\le B_0$.
Counting the subsets with sum exactly $B_0$ is \#P-hard, because the NP-hardness reduction for this problem~\cite{Kar72} can be made parsimonious~\cite{Sim77}.
That count is the number of subsets with sum at most $B_0$ minus the number with sum at most $B_0-1$, where the latter is $0$ if $B_0=0$.
So \#Knapsack is \#P-hard as well.
Let $K_{\mathrm{NO}}$ be the number of subsets with sum larger than~$B_0$; the answer is $2^k-K_{\mathrm{NO}}$.
Set
\[
  \bar k=k+1,\qquad h=2S,\qquad B=h+B_0,\qquad Q=64\bar k^2B,\qquad q=\frac1Q .
\]
For $L\in\{L_-,L_+\}$ with $L_-=B+\tfrac14$ and $L_+=B+\tfrac12$, the instance $I_L$ has two machines and the following $k+3$ jobs:
\begin{itemize}
\item for each $i\in[k]$, a job of size $a_i$ with probability $1/2$;
\item a job $H$ of size $h$ with probability $1/2$;
\item a job $b$ of size $L$ with probability $q$;
\item a deterministic job $d$ of size $B+1$.
\end{itemize}
The job $b$ has mean $qL<1/2$, which is smaller than every other mean.
So SEPT starts $b$, then the jobs $a_i$ in increasing order of size, then $H$, then~$d$.

\begin{lemma}\label{lem:blockerfirst}
In the instances $I_{L_\pm}$, consider a state at time~$0$ in which $b$ has not started and a machine is free at time~$0$.
Some optimal continuation from this state is non-idling and starts $b$ next, at time~$0$.
\end{lemma}

\begin{proof}
We use induction on the number of remaining jobs; the case of one job is trivial.
By Lemma~\ref{lem:nonidle}, some optimal continuation is deterministic and non-idling, so it starts its first job at time~$0$.
Every job other than~$b$ is positive with probability at least $1/2$, and every positive processing time is at least~$1$.

Suppose exactly one machine is free at time~$0$, and a non-idling optimal continuation starts a job $i\ne b$ there.
The other machine runs a job that started at time~$0$ and has at least one time unit left.
So if $i$ is positive, both machines are busy until time~$1$, and $\E[S_b]\ge1/2$.
Let $c_0$ be the minimum expected sum of start times for the remaining jobs other than~$b$ from the current state, allowing idle time.
Simulating $b$ with an independent virtual outcome, as in Lemma~\ref{lem:basic}(b), shows that the continuation's start times of the jobs other than~$b$ sum to at least $c_0$ in expectation.
So its expected sum of start times is at least $c_0+1/2$.
Now start $b$ instead.
If $b$ is zero, continue optimally, with expected sum of start times~$c_0$.
If $b$ is positive, which happens with probability~$q$, run the remaining jobs one after another.
At most $k+3\le2\bar k$ jobs remain, each of size at most $B+1\le2B$.
So every start time is at most $4\bar kB$, and the sum of start times is at most $8\bar k^2B$.
This alternative has expected sum of start times at most $c_0+q\cdot8\bar k^2B=c_0+1/8<c_0+1/2$, a contradiction.
So every non-idling optimal continuation starts $b$ at time~$0$.

Suppose both machines are free at time~$0$, and a non-idling optimal continuation starts $i\ne b$ at time~$0$.
After $i$ starts, a machine is still free at time~$0$, so by induction we may replace the continuation by one that starts $b$ next, at time~$0$.
Exchanging the order of $i$ and $b$ leaves both start times equal to~$0$.
It also leads to the same remaining jobs and the same multiset of machine release times in every realization, so the cost is unchanged.
\end{proof}

\begin{proof}[Proof of Theorem~\ref{thm:hard}]
By Lemma~\ref{lem:blockerfirst}, some optimal policy starts $b$ at time~$0$.
Afterwards the machines become available at times $(0,0)$ if $b$ is zero and $(0,L)$ if $b$ is positive.
The remaining jobs are stochastically ordered.
A job of size $a_i$ with probability $1/2$ is dominated by one of size $a_{i'}\ge a_i$ with the same probability.
The job $H$ dominates all of them because $h=2S>a_i$.
The job $d$ dominates $H$ because $B+1>h$.
By Lemma~\ref{lem:wvw}, SEPT is optimal from either state, so SEPT is optimal for $I_{L_-}$ and~$I_{L_+}$.
So computing $\OPT$ and computing the cost of SEPT coincide on these instances.

We compare the two instances under SEPT, with the same outcomes for all jobs.
If $b$ is zero, the two schedules are identical.
Suppose $b$ is positive, and let $W=\sum_ia_iZ_i$ be the total realized processing time of the jobs $a_1,\dots,a_k$, where $Z_i$ are their outcome indicators.
Since $W\le S<B<L_-$, all these jobs run on the free machine before time~$L$, and $H$ starts there at time~$W$.
If $H$ is zero, $d$ starts at time $W$.
If $H$ is positive and $W\le B_0$, then $H$ completes at $W+h\le B<L_-$, and $d$ starts at time $W+h$.
If $H$ is positive and $W>B_0$, then $W\ge B_0+1$, so $H$ runs until $W+h\ge B+1>L_+$, and $d$ starts at time~$L$ on the other machine.
Only in the last case do the start times differ, by $L_+-L_-=1/4$.
The completion time of $b$ itself also differs by $1/4$ when $b$ is positive.
Hence
\[
  \OPT(I_{L_+})-\OPT(I_{L_-})=\frac q4\Bigl(1+\frac12\cdot\frac{K_{\mathrm{NO}}}{2^k}\Bigr)=\frac q4\Bigl(1+\frac{K_{\mathrm{NO}}}{2^{\bar k}}\Bigr),
\]
and two exact evaluations determine $K_{\mathrm{NO}}$.
All numbers have bit length polynomial in the input.
Multiplying all sizes by~$4$ makes them integers and does not change the argument.

For the decision version, note that all processing times are multiples of $1/4$ and the probability of every joint outcome is a multiple of $1/(Q2^{\bar k})$.
So $\OPT$ is a multiple of $1/(4Q2^{\bar k})$.
It is also at most $8\bar k^2B+(k+3)(B+1)$.
Binary search with polynomially many threshold queries therefore recovers $\OPT$ exactly.
\end{proof}

\subsection{Approximate evaluation}

\begin{theorem}\label{thm:fpras}
Let $\Pi$ be a non-idling policy whose decisions can be simulated in polynomial time.
For $\rho,\delta\in(0,1)$, an estimate $\hat C$ with $(1-\rho)\cost(\Pi)\le\hat C\le(1+\rho)\cost(\Pi)$ with probability at least $1-\delta$ can be computed from $\lceil3N\rho^{-2}\ln(2/\delta)\rceil$ simulations of $\Pi$.
\end{theorem}

\begin{proof}
For a realization $\omega$, let $V(\omega)=\sum_jX_j(\omega)$ and $C(\omega)=\sum_jC^\Pi_j(\omega)$.
A non-idling policy completes every job by time $V(\omega)$, so $0\le C(\omega)\le N\,V(\omega)$, and $C(\omega)=0$ if $V(\omega)=0$.
Sample realizations from the following distribution~$\mathbb Q$.
Choose a job $i$ with probability $\mu_i/P$, set $X_i=p_i$, and draw all other outcomes from their own distributions.
The policy sees only the outcomes it observes when starting jobs.
Under~$\mathbb Q$, the density of a realization with respect to the true distribution is
\[
  \sum_i\frac{\mu_i}{P}\cdot\frac{\mathbf 1[X_i=p_i]}{q_i}=\frac{V(\omega)}{P}.
\]
Hence the estimator $Z=P\,C/V$ satisfies $\E_{\mathbb Q}[Z]=\cost(\Pi)$.
It takes values in $[0,NP]$, and its mean is at least $P$.
The multiplicative Chernoff bound for the average of $\lceil3N\rho^{-2}\ln(2/\delta)\rceil$ independent copies of $Z/(NP)\in[0,1]$, whose mean is at least $1/N$, gives the claim.
\end{proof}

The same argument applies to every policy whose realized sum of completion times is at most $\beta N\,V(\omega)$ for a constant $\beta$, with the number of simulations multiplied by~$\beta$.
For instance, $\beta=\gamma$ for the grid-non-idling reuse-rule policies of Section~\ref{sec:dp}.

\section{List policies lose a logarithmic factor}\label{sec:barriers}

A \emph{list policy} fixes an order of the jobs before any outcome is observed and starts the jobs in this order.
The order may depend on the instance and on random coins.
Eberle, Fischer, Matuschke, and Megow~\cite{EFMM19} show that index policies, which rank each job by its own distribution, are not constant-factor approximations.
Im, Moseley, and Pruhs~\cite{IMP15} give a list policy with ratio $O(\log^2N+m\log N)$, which is $O(\log^2N)$ on two machines.
Gupta, Moseley, and Zhou~\cite{GMZ23} compute a list policy whose ratio is $\widetilde O(\sqrt m)$, where $\widetilde O$ hides factors polylogarithmic in~$N$.
The following theorem shows that a dependence on $N$ is necessary for every list policy, even on two machines and even if the policy may idle.
So the bound of~\cite{IMP15} for two machines is tight up to a factor $O(\log N)$.

\begin{theorem}\label{thm:list}
For every integer $b\ge2$, let $M=2^b$ and consider two machines with $M$ jobs of each type $i=0,\dots,b-1$, where a job of type $i$ has
\[
  p_i=\frac{M^2}{4^i}\qquad\text{and}\qquad q_i=\frac{2^i}{M^2}.
\]
Let $N=bM$.
Every list policy costs at least $Mb^2/128$, whereas $\OPT\le6Mb$.
Hence every list policy has ratio at least $b/768\ge\log_2N/1536$.
\end{theorem}

The sizes are the distinct integers $4^{b-i}\ge4$, and $Mq_i\le1/2$.
The means $\mu_i=2^{-i}$ decrease rapidly with the type, while $p_iq_i^2=M^{-2}$ does not depend on the type.
Every policy has $P=\sum_iM2^{-i}=2M-2$.

\begin{proof}[Proof of the lower bound]
Fix a deterministic list; for a random list, the bound follows by averaging over the coins.
For $0\le k<N$, let $F_k$ be the start time of the $(k+1)$-st job of the list, let $k_i$ be the number of type-$i$ jobs among the first $k$, and let $B_i\sim\mathrm{Bin}(k_i,q_i)$ be the number of positive ones among them.
At time $F_k$ all these $k$ jobs have started and a machine is free, so at most one of them is still running.
So at least $(B_i-1)_+$ positive jobs of type $i$ have completed for each $i$, and since two machines process at most $2F_k$ units of work by time $F_k$,
\[
  2F_k\;\ge\;\sum_{i=0}^{b-1}p_i(B_i-1)_+ .
\]
For $k_i\ge2$, keeping only the event that exactly two jobs are positive gives
\[
  \E\bigl[(B_i-1)_+\bigr]\;\ge\;\binom{k_i}{2}q_i^2(1-q_i)^{k_i-2}\;\ge\;\frac14k_i(k_i-1)q_i^2,
\]
because $(1-q_i)^{k_i-2}\ge1-Mq_i\ge1/2$.
The bound also holds for $k_i\le1$, as both sides are then zero.
Using $p_iq_i^2=M^{-2}$ and $\sum_i(k_i-k/b)^2\ge0$,
\[
  \E[F_k]\;\ge\;\frac1{8M^2}\sum_ik_i(k_i-1)\;\ge\;\frac1{8M^2}\Bigl(\frac{k^2}{b}-k\Bigr).
\]
For $N/2\le k<N$ we have $k\ge bM/2\ge2b$, so $k^2/b-k\ge k^2/(2b)\ge bM^2/8$ and $\E[F_k]\ge b/64$.
There are $N/2$ such values of~$k$, so the cost is at least $P+(N/2)(b/64)\ge Mb^2/128$.
\end{proof}

\begin{proof}[Proof of the upper bound]
Consider the following policy.
On machine~$1$, start jobs in increasing order of type until the first positive outcome; all earlier jobs are zero and complete at time~$0$.
If all jobs are zero, every job completes at time~$0$.
Otherwise, let $i$ be the type of the first positive job.
Leave that job on machine~$1$, and process all remaining jobs one after another on machine~$2$, starting at time~$0$, in decreasing order of type.

Let $E_i$ be the event that the first positive job has type~$i$.
It requires a positive job among the $M$ jobs of type~$i$, so $\Pr[E_i]\le Mq_i=2^i/M$.
Condition on the history up to the first positive job.
The types below~$i$ are exhausted, and the remaining jobs keep their original distributions.
On a single machine, the processing time of a job is added to the start time of each later job.
In decreasing order of type, a job of type $j\ge i$ has at most $M(j-i+1)$ later jobs, and there are at most $M$ jobs of type~$j$, each of mean~$2^{-j}$.
So the conditional expected sum of start times is at most
\[
  M^2\sum_{j\ge i}(j-i+1)2^{-j}\;\le\;M^22^{-i}\sum_{t\ge0}(t+1)2^{-t}\;=\;4M^22^{-i}.
\]
Averaging over the disjoint events $E_i$, the expected sum of start times is at most $\sum_i(2^i/M)\cdot4M^22^{-i}=4Mb$.
So the cost is at most $P+4Mb\le6Mb$.
\end{proof}

The ratio follows from $\OPT\le6Mb$ and $\log_2N=b+\log_2b\le2b$.
The adaptive policy above changes its order after the first positive outcome.

\section{Open problems}\label{sec:open}

The main open question is whether Bernoulli jobs admit a polynomial-time approximation scheme, or at least a polynomial-time constant-factor approximation, when $m$ is part of the input.
After the size reduction, the dynamic program of Section~\ref{sec:dp} is quasi-polynomial for two reasons.
It records the number of remaining jobs in each of the $O(\varepsilon^{-1}\log(N/\varepsilon))$ size classes, and it records the release times as a multiset over $O(\varepsilon^{-2}\log(N/\varepsilon))$ candidate values.
A polynomial-time scheme needs a more compact description of both, and by Theorem~\ref{thm:list} it must be adaptive.

A natural alternative is to round a relaxation over start-time distributions.
For Bernoulli jobs, such a relaxation can include, for every deadline~$u$, the expected-rank constraints of the laminar matroid that allows at most $m(\lfloor u/h\rfloor+1)$ positive jobs of size at least $h$ to start by time~$u$.
Constraints that act on one deadline at a time are not enough.
On one machine, take jobs with $p_i=2^{N-i+1}$ and $q_i=1/p_i$ for $i=1,\dots,N$, so that $\OPT=N(N+1)/2$.
Start job $i$ at time~$0$ with the probability that no larger job is positive, and at time $p_i$ otherwise.
This satisfies every such rank constraint, because the first positive job of size larger than $u$ together with all positive smaller jobs forms an independent set.
It also keeps the expected processing in every interval below the length of the interval, because $\sum_iq_i<1$.
Its objective is below~$2N$.
On one machine, the inequality $\sum_j\mu_j\E[S_j]\ge\sum_{i<j}\mu_i\mu_j$ of M{\"o}hring, Schulz, and Uetz~\cite{MSU99} closes this gap, but its analog for $m$ machines contains a term proportional to~$\Delta$.
Whether a relaxation that combines both kinds of constraints has a constant integrality gap for Bernoulli jobs is open.

Two intermediate questions seem natural.
The first is a polynomial-time approximation scheme for every fixed number of machines, that is, in time $N^{f(m,\varepsilon)}\cdot\mathrm{poly}(\mathcal L)$.
The second is a polynomial-time approximation whose ratio is $o(m)$ and independent of~$N$.

\section*{AI usage disclosure}
\addcontentsline{toc}{section}{AI usage disclosure}

AI agents were used for exploratory discussions during the development of the proofs and to assist with drafting and revising the manuscript.
The author takes full responsibility for the paper's contents and correctness.

\bibliographystyle{alpha}
\bibliography{refs}

\end{document}